\documentclass[a4paper, accepted=2023-04-28]{quantumarticle}
\pdfoutput=1

\usepackage{amsmath, amsfonts, amssymb, bbm, physics}
\usepackage{multicol, color, enumitem}
\usepackage[hidelinks]{hyperref}

\usepackage{geometry}
\usepackage{amsthm, thmtools}
\newtheorem{theorem}{Theorem}[section]
\newtheorem{lemma}{Lemma}[section]

\usepackage[citestyle=numeric-comp, sorting=none, doi=true,isbn=false,url=false,eprint=false]{biblatex}
\usepackage{authblk}

\title{Only Classical Parameterised States have Optimal Measurements under Least Squares Loss}
\author[1,2,*]{Wilfred Salmon}
\author[1]{Sergii Strelchuk}
\author[2]{David R. M. Arvidsson-Shukur}

\affil[1]{DAMTP, Centre for Mathematical Sciences, University of Cambridge, Cambridge CB30WA, UK}
\affil[2]{Hitachi Cambridge Laboratory, J. J. Thomson Avenue, CB3 0HE, Cambridge, UK}
\affil[*]{was29@cam.ac.uk}

\def\th{\ensuremath{\hat{\theta}}}
\def\ttd{\ensuremath{\tilde{\theta}}}
\def\tb{\ensuremath{\hat{\theta}^\text{B}}}

\def\Tr{\ensuremath{\text{Tr}}}

\def\st{\ensuremath{\text{s.t.}}}
\def\rb{\ensuremath{\bar{\rho}}}
\def\Kperp{\ensuremath{{K^\perp}}}

\begin{document}

\maketitle

\begin{abstract}
Measurements of quantum states form a key component in quantum-information processing. It is therefore an important task to compare measurements and furthermore decide if a measurement strategy is optimal. Entropic quantities, such as the quantum Fisher information, capture asymptotic optimality but not optimality with finite resources. We introduce a framework that allows one to conclusively establish if a measurement is optimal in the non-asymptotic regime. Our method relies on the fundamental property of expected errors of estimators, known as risk, and it does not involve optimisation over entropic quantities. The framework applies to finite sample sizes and lack of prior knowledge, as well as to the asymptotic and Bayesian settings. We prove a no-go theorem that shows that only classical states admit optimal measurements under the most common choice of error measurement: least squares. We further consider the less restrictive notion of an approximately optimal measurement and give sufficient conditions for such measurements to exist. Finally, we generalise the notion of when an estimator is inadmissible (i.e. strictly worse than an alternative), and provide two sufficient conditions for a measurement to be inadmissible.
\end{abstract}

\section{Introduction}
Extracting information from systems requires performing a measurement. The outcome of the measurement is then used to infer information about some property of the system of interest. Parameter estimation is the well-established study of different strategies of measurement and the subsequent inference of some number of parameters of a system. In particular, when the underlying system is quantum then understanding the non-classical features of quantum measurement becomes of paramount importance and forms the central object of study in quantum metrology. 

The asymptotic regime of such problems, corresponding to an infinite number of measurements, has been extensively studied \cite{Alba2020, Gore2022}. This regime is also known as local estimation, as one only considers infinitesimal changes in parameter values. The quantum Fisher information \cite{Brau1994, Alba2020} has emerged as the standard metric of asymptotic information for a parametrised quantum state. Its operational significance is emphasised by the quantum Cram{\'e}r-Rao bound \cite{Brau1994, Alba2020} where its inverse lower-bounds the variance of locally unbiased estimators. However, asymptotically optimal measurements, as characterised by the quantum Fisher information, may be sub-optimal in finite-resource settings (see below).
Additionally, in practice, parameter estimation is a global problem; one must consider parameters across the whole parameter space. Global estimation is often studied within the Bayesian regime, where one has some \textit{a priori} belief about the distribution of the parameter. Optimal measurement choices in the Bayesian regime are well understood \cite{Demk2020, Fuch2009}, however their optimality relies on an accurate choice of prior. To the best of our knowledge, a thorough analysis of measurement optimality in the most general setting of global parameter estimation without prior knowledge is missing. 

In this paper, we provide such an analysis. Our methodology relies on the fundamental property of expected errors of estimators, known as risk, defined below. We establish a fundamental result of quantum-information processing: \textit{only} classical parametrised states \cite{Alba2020}, those that are a classical mixture of some parameter-independent pure states, admit optimal measurements. We go on to propose three ways in which a measurement can be approximately optimal, and we prove that if a state is close to being classical, there exists an approximately optimal measurement. Finally, we consider when one measurement dominates (is strictly better than) another measurement. We present two ``bad" classes of measurements---those that extract no information from a system and those that can be fine-grained. We show that these measurements are always dominated by suitable different measurements, which we construct.

\section{Technical Background}
We begin by reviewing classical parameter estimation \cite{Lehm1998}. As an example, suppose one wishes to conduct an experiment to estimate $q\in[0,1]$, some unknown proportion of a population satisfying some property. We sample some subset of $n$ members of the population and observe a realisation of a binomially distributed random variable $X\sim B(n,q)$. We do not, \textit{a priori}, know $X$'s distribution, only that it will be binomially distributed according to the parameter $q$. Based off a realisation $x$ of $X$, one gives an estimate of $q$, most commonly $x/n$. In the general setting, one wishes to estimate some parameter $\theta$, which could be any of the elements of the parameter space $\Theta$. We observe the outcome of some random variable $X$, whose distribution is unknown to us, but is contained in the set $\{P_\theta : \theta\in\Theta\}$. An estimator of $\theta$ is then a function $\th(x)$, where $\th(x)$ encodes our estimate of the underlying parameter $\theta$ in the case that $X$ takes the value $x$.

To discuss the merit of estimators, one introduces a loss function $L$, where $L(\theta_1,\theta_2)$ quantifies how bad a guess of $\theta_1$ if the true parameter were $\theta_2$. Common examples include least squares $L(\theta_1,\theta_2) = ||\theta_1-\theta_2||^2$, the Kullback-Leibler divergence \cite{Kull1951} and the Bhattacharyya distance \cite{Bhat1946}. The loss function induces a risk function for an estimator:
\begin{equation}
R(\th,\theta) = \mathbb{E}_{X\sim P_\theta}[L(\th(X),\theta)].
\end{equation}
$R(\th,\theta)$ measures how badly one expects an estimator to perform for a fixed value of the underlying parameter $\theta$. An estimator is ``good" at the point $\theta$ if it has a low risk there. We say that $\th$ dominates $\ttd$ if $\th$ is always at least as good as $\ttd$ and sometimes better, i.e.,
\begin{align}
&\forall \theta\in\Theta,\; R(\th, \theta)\leq R(\ttd, \theta). \notag\\
\text{and}\;\; &\exists \theta\in\Theta\;\st \; R(\th,\theta)<R(\ttd, \theta). \label{eqn:dominates}
\end{align}
If $\ttd$ is dominated by some estimator $\th$ one says that it is inadmissible, otherwise it is said to be admissible. Admissibility is a very weak condition; constant estimators---always guessing $\theta_0$ regardless of the outcome of the experiment---are often admissible \cite{Wass2004}, since they have zero risk at $\theta_0$. However, admissibility is certainly desirable, or, by definition, one can choose a nominally ``better" estimator.

In the quantum setting \cite{Hels1969, Giov2011}, our parameter is encoded in some quantum state $\rho(\theta)$. Often, one may have $n$ copies of the same state $\sigma(\theta)$ in which case $\rho(\theta)=\sigma(\theta)^{\otimes n}$. We choose some (parameter independent) generalised measurement $M = \{M_k\}$, where each $M_k$ is a linear operator satisfying $M_k\geq 0$, and $\sum_k M_k = \mathbbm{1}$. This gives rise to a probability distribution parametrised by $\theta$:
\begin{equation}\label{eqn:prob_distr}
p_k^{(M)}(\theta) = \Tr[\rho(\theta)M_k].
\end{equation}
Thus, after measurement, the quantum experiment reduces to a classical parameter-estimation problem and one can pick an estimator $\th_M(k)$ using the probability distribution in Eq. \eqref{eqn:prob_distr}. Hence, the objects of interest in quantum parameter estimation are pairs $(M,\th_M)$, where $M$ is a generalised measurement and $\th_M$ is an estimator based off the outcome of $M$. Note that an experimenter may choose a POVM without any particular estimator $\th_M$ in mind. However when the experimental data is eventually used to estimate the parameter, the way in which the data is processed exactly corresponds to picking an estimator. Often, when it is clear from context, we will omit the $M$ label of the estimator. We can compare the performance of estimators based on different measurements. In particular, for two pairs $(M,\th_M)$ and $(F,\ttd_F)$ we write $\th_M\leq\ttd_F$ to mean that $\th_M$ is always at least as good as $\th_F$, i.e. $\forall \theta\in\Theta,\; R(\th_M,\theta)\leq R(\ttd_F,\theta)$. We say that $\th_M$ dominates $\ttd_F$, written $\th_M < \ttd_F$, if $\th_M\leq\ttd_F$ and $\th_M$ is sometimes better, i.e. $\exists \theta\in\Theta\;\st \; R(\th_M,\theta)<R(\ttd_F, \theta)$, in close analogy to equation \eqref{eqn:dominates}.

Before we introduce our notion of measurement optimality, we briefly discuss how the commonly used notion, based on the (quantum) Fisher information $\mathcal{F}(\theta)$ \footnote{For the ``single parameter" case $\Theta\subseteq\mathbb{R}$, the quantum Fisher information $\mathcal{F}(\theta)$ is defined as $\Tr(\rho(\theta)\Lambda(\theta)^2)$. Here, $\Lambda(\theta)$ is the symmetric logarithmic derivative, defined implicitly by the equation $2\partial_\theta \rho(\theta) = \Lambda(\theta)\rho(\theta) + \rho(\theta)\Lambda(\theta)$. For the ``multiparameter case" of $\Theta\subseteq\mathbb{R}^N$ with $N>1$ there are multiple definitions of the quantum Fisher information; for a review see \cite{Alba2020}.}, is, in general, unsatisfactory. $\mathcal{F}(\theta)$  is an entropic quantity that characterises the performance of estimation protocols in the asymptotic limit:  asymptotic risk is lower-bounded by $\Tr(\mathcal{F}(\theta)^{-1})$. We now give two examples where $\mathcal{F}(\theta)$ does not capture optimality. First, consider estimating the mean of three or more normally distributed variables. The most common choice of estimator, the maximum likelihood estimator, has risk everywhere equal to $\Tr(\mathcal{F}(\theta)^{-1})$. Thus, based on the Fisher information, it is asymptotically optimal. However, the James-Stein estimator \cite{Stei1961} always has strictly smaller risk (it dominates the MLE). Secondly, suppose one can query a unitary $U(\theta) =$ diag$\{1, e^{i\theta}\}$ a number $N$ times to estimate $\theta$. 
A strategy that maximises $\mathcal{F}(\theta)$,  applies the unitary $N$ times to the $\ket{+}$ state and measures in the $\{\ket{\pm}\}$ basis \cite{Lee2002}. However, in this scenario, one only receives one bit of information so that the risk of $\th$ does not scale with $N$. In the optimal case, (e.g. the quantum phase estimation algorithm \cite{Neil2010}) the risk scales as $1/N$.\\
 
\section{Optimal Measurements}
Below, we extend the aforementioned framework for estimator comparisons to include quantum measurements. Here, we define what we argue is the only possible such extension, unless further modelling assumptions are made. We begin by extending the notion of admissibility of estimators to measurements.  We say that a generalised measurement $M$ is at least as good as a generalised measurement $F$, written $M\preceq F$ if 
\begin{equation}
\forall \th_F\; \exists \ttd_M\; \st\; \ttd_M\leq\th_F.
\end{equation}
If $M\preceq F$, then if one swaps an experiment's measurement from $F$ to $M$, then whatever estimation strategy $\th_F$ was in use before, one can pick a new one $\th_M$ that is at least as good as $\th_F$ regardless of the true underlying parameter. If, however, $M\npreceq F$, then one cannot do so---by swapping to $M$, there will be at least one estimation strategy $\th_F$ such that whatever new strategy $\th_M$ one picks there are some values of the underlying parameter where one would expect to do worse. Thus, our definition is the natural and only way one can define $M$ as being at least as good as $F$ without further assumptions on the estimation problem. We say that $M$ dominates $F$, written $M\prec F$, if $M\preceq F$, but $F\npreceq M$. In analogue with the classical case, we say that a measurement $M$ is admissible if no other measurement dominates it. Clearly, as in the classical case, admissibility of a measurement is strongly desirable. 

There is then a natural definition of an optimal measurement:  we say that $M$ is optimal if for all other possible measurements $F$, $M\preceq F$. This is the strongest and most general definition of optimality one could make---we emphasise that any other definition must make some additional modelling assumptions (e.g., some \textit{a priori} belief about the unknown parameter). Nonetheless, and somewhat surprisingly, we show that optimal measurements exist for a family of parameter estimation problems.

\begin{theorem}\label{thm:clas_optimal}
	Let $\Theta\subseteq\mathbb{R}^N$ and let $L$ be a choice of loss function that is convex in its first argument. Suppose that $\rho(\theta) = \sum_i p_i(\theta)\dyad{i}$, for some probability distribution $\{p_i(\theta)\}$ and a fixed, $\theta$-independent, basis $\{\ket{i}\}$, then the measurement $M=\{\dyad{i}\}$ is optimal.
\end{theorem}
\begin{proof}
	Suppose $F=\{F_k\}$ is some other measurement with an estimator $\ttd_F$. Let $m_{k,i}=\Tr(F_k\dyad{i})$. Note that $m_{k,i}\geq0,\sum_k m_{k,i}=1$ and $\sum_i p_i(\theta)m_{k,i} = p_k^{(F)}(\theta)$. Define an estimator using $M$ by $\th_M(i)=\sum_k m_{k,i}\ttd_F(k)$. Then
	\begin{align}
	R(\th_M,\theta) &= \sum_i p_i(\theta)L\left(\sum_k m_{k,i}\ttd_F(k), \theta\right),\\
	&\leq \sum_{i,k}p_i(\theta)m_{k,i}L(\ttd_F(k),\theta) = R(\ttd_F,\theta).
	\end{align}
	Consequently, by swapping $\ttd_F$ to $\th_M$ we can only decrease the risk and thus $M\preceq F$.
\end{proof}
 
 Note that the parameter may be a vector ($N>1$), so that Theorem \ref{thm:clas_optimal} holds in the multi-parameter setting. States where the eigenbasis of the quantum state may be expressed in a parameter independent way, are called classical parametrised states \cite{Alba2020}. More precisely, a parametrised quantum state $\rho(\theta)$ is called classical if there exists a $\theta$-independent eigenbasis $\{\ket{i}\}$ and probabilities $\{p_i(\theta)\}$ such that $\rho(\theta) =\sum_i p_i(\theta)\dyad{i}$. For example, classical states are relevant for estimating temperature, where the system is in a (thermal) state:

\begin{equation}
\rho(\beta) = \frac{e^{-\beta H}}{\Tr(e^{-\beta H})}.
\end{equation}
Theorem \ref{thm:clas_optimal} shows that  measuring the energy of the system is always an optimal strategy, regardless of the ensuing choice of estimator.\\

As a second example, we consider estimating the strength of depolarizing noise acting on a fixed quantum state. Suppose that depolarising noise acts on some $d$-dimensional pure state $\ket{\psi}$ with probability\footnote{The full range of $p$ for which this is a valid state is $p\in[0,d/(d-1)]$. However, for $p>1$ the state does not represent realistic noise models.} $p\in[0,1]$ to produce the mixed state:
\begin{equation}
	\rho(p)=(1-p)\dyad{\psi}+ p\frac{\mathbbm{1}}{d}.
\end{equation}
Theorem \ref{thm:clas_optimal} shows that measuring in any basis containing $\ket{\psi}$ is optimal for estimating $p$.

It is natural to ask whether there exist any parametrised quantum states that are not classical, but still admit an optimal measurement. We give a partial converse to Theorem \ref{thm:clas_optimal}, in the case that the loss function is least squares.

\begin{restatable}{theorem}{cto}\label{thm:converse_to_optimality}
	Suppose $\rho(\theta)$ is some parameter estimation problem, where $\Theta\subseteq \mathbb{R}^N$ is convex. Unless $\rho(\theta)$ is classical, there is no optimal measurement $M$ under least-squares loss.
\end{restatable}
 The proof of this Theorem is given in Appendix \ref{sec:conv_to_opt}. However, here we provide a sketch of the proof. 
\begin{proof}[Proof Sketch]
We assume the existence of an optimal measurement $M$ and deduce that $\rho(\theta)$ must be classical. We first restrict our analysis to the single-parameter ($N=1$) case. We fix $\theta_1,\theta_2\in\Theta$ and consider the restricted parameter space $\{\theta_1,\theta_2\}\subseteq\Theta$. We then consider the Bayesian scenario of being given the states $\rho(\theta_1)$ with probability $p$ and $\rho(\theta_2)$ with probability $1-p$. We show that if an optimal measurement exists for this Bayesian parameter-estimation problem, then $\rho(\theta_1)$ and $\rho(\theta_2)$ commute. However, we also show that $M$ must be optimal for any such Bayesian problem. Thus, we deduce that if an optimal measurement exists on $\Theta$, all $\{\rho(\theta):\theta\in\Theta\}$ must commute. We use this to simultaneously diagonalise the $\{\rho(\theta)\}$ and thus we deduce that $\rho$ is classical. By restricting to lines in higher dimension parameter spaces, we generalise the single parameter case to the multiparameter case.
\end{proof}

Thus, for $\rho(\theta)$ non-classical, there is no optimal measurement strategy. Whatever measurement $M$ one picks, there is always a measurement estimator pair $(F,\th_F)$ that one cannot guarantee to do at least as well as. Some trade-off or additional modelling assumption is required to decide on the best measurement strategy.\\

To develop physical intuition for Theorems \ref{thm:clas_optimal} and \ref{thm:converse_to_optimality}, it is useful to consider distinguishablity of states. For $\sigma$ and $\nu$ quantum states, the optimal measurement to distinguish them \cite{Hels1969} is with respect to the eigenbasis of $\sigma - \nu$. If $\rho(\theta)$ is a classical parametrised state, then for any $\theta_1, \theta_2\in\Theta$ the eigenbasis of $\rho(\theta_1)-\rho(\theta_2)$ is constant and thus the same measurement is optimal for distinguishing all possible pairs of states. Since it is pairwise optimal, it is intuitive that this measurement is then ``globally" optimal for parameter estimation. For non-classical states, different measurements are better at distinguishing different parts of the parameter space. This effect is best demonstrated by the ``most non-classical" parametrised state: a pure state.
For $\theta\in[0,2\pi)$, the Mach-Zehnder state \cite{Giov2011} is
\begin{equation}
	\ket{\psi(\theta)} = \frac{1}{\sqrt{2}}(\ket{0} + e^{i\theta}\ket{1}).
\end{equation}

Consider a measurement $M$ of the state defined by measuring with respect to the basis $\{\ket{+},\ket{-}\}$. $M$ saturates the quantum Fisher information  (for all $\theta$) \cite{Giov2011} and is thus a good candidate for an optimal measurement. It is easy to check that the outcome probability distribution is
\begin{equation}\label{eqn:pm_probs}
	p_+(\theta) = \cos^2(\theta/2), \quad p_-(\theta) = \sin^2(\theta/2).
\end{equation}

Consider a different measurement $F$, defined by measuring with respect to the basis
\begin{align}
	\ket{e_1} &= \frac{1}{\sqrt{2}}(\ket{0} + e^{i\pi/4}\ket{1}), \nonumber\\
	\ket{e_2} &= \frac{1}{\sqrt{2}}(e^{-i\pi/4}\ket{0} - \ket{1}).
\end{align}

Measuring with respect to $F$ gives rise to a probability distribution of
\begin{align}
	p_1(\theta) &= \cos^2\left[\frac{1}{2}\left(\theta-\frac{\pi}{4}\right)\right],\nonumber \\
	p_2(\theta) &= \sin^2\left[\frac{1}{2}\left(\theta-\frac{\pi}{4}\right)\right].
\end{align}

Define an estimator $\th_F$ by $\th_F(1) = \pi/4,\ \th_F(2) = \pi/2$. Note in particular that $\th_F$ has zero risk at $\pi/4$ and thus we have a non-constant estimator that has zero risk at $\pi/4$. If $M$ were optimal there would exist some estimator $\th_M$ such that $\th_M\leq\th_F$. But then we need $\th_M$ to have zero risk at $\pi/4$, but from equation \eqref{eqn:pm_probs}, we see this is only possible with the constant estimator $\th_M(+)=\th_M(-)=\pi/4$. However, in this case $R(\th_F,\pi/2) < R(\th_M, \pi/2)$ and we deduce that $M\npreceq F$. Thus $M$ is certainly not optimal.

Note that $M$ and $F$ both have maximal Fisher information and thus the measurements both have the same asymptotic, local, estimation ability. However, within global estimation they are \textit{not} interchangeable; $M$ cannot discriminate between $[0,\pi )$ and $[\pi, 2\pi )$ whereas $F$ cannot discriminate between $[\pi/4,5\pi/4 )$ and $[5\pi/4,9\pi/4)$. Thus, given a large number of copies it would be a poor strategy to measure all of the states with $M$ or all of the states with $F$ as one could not distinguish $\theta$ from $\theta + \pi$. Thus in practice, one would use a collection of different measurements; see for example \cite{Smit2022}.\\

The above argument generalises to other pure states. Suppose one has some (generically multiparameter) parametrised pure state $\ket{\psi(\theta)}$. Then for any fixed $\theta_*$, letting $F_{\theta_*}$ be a measurement in some basis containing $\ket{\psi(\theta_*)}$, one can clearly construct an estimator with zero risk at $\theta_*$. However, in general, the only estimators $\th_{F_{\theta_*}}$ with zero risk at some $\theta\neq\theta_*$ will be constant. Thus we expect none of the $\{F_{\theta_*}|\theta_*\in\Theta\}$ to be optimal and thus expect no optimal measurement for $\ket{\psi(\theta)}$.

\section{Approximately Optimal Measurements}\label{sec:approx_opt}
Since the previous section shows that optimal measurements (under least-squares loss) only exist for the very restrictive class of classical parametrised states, it is natural to ask whether measurements can be approximately optimal. There are three clear candidates for defining when a measurement is approximately optimal:
\begin{enumerate}[label=(\roman*)]
	\item{
		A measurement $M$ is $\epsilon$-additively optimal if, for any other measurement-estimator pair ($F$, $\th_F$) there exists an estimator $\th_M$ such that for all $\theta\in\Theta$, 
		\begin{equation}
		R(\th_M,\theta)\leq R(\th_F,\theta)+\epsilon.
		\end{equation}
	}
	\item{
		A measurement $M$ is $\eta$-multiplicatively optimal if, for any other measurement-estimator pair ($F$, $\th_F$) there exists an estimator $\th_M$ such that for all $\theta\in\Theta$,
		\begin{equation}
		R(\th_M,\theta)\leq (1+\eta)R(\th_F,\theta).
		\end{equation}
	}
	\item{
		A measurement is $\delta$-locally optimal if, for any other measurement-estimator pair ($F$, $\th_F$) there exists an estimator $\th_M$ and subset $S\subseteq\Theta$ such that $S$ has volume (measure) less than or equal to $\delta$ and for all $\theta$ in $\Theta\setminus S$,
		\begin{equation}
		R(\th_M,\theta)\leq R(\th_F,\theta).
		\end{equation}
	}
\end{enumerate}
We can combine the third definition with either of the first two. For example, a measurement is $\epsilon$-additively and $\delta$-locally optimal if, for any other measurement-estimator pair ($F$, $\th_F$) there exists an estimator $\th_M$ and subset $S\subseteq\Theta$ such that $S$ has volume (measure) less than or equal to $\delta$ and for all $\theta$ in $\Theta\setminus S$,
\begin{equation}
R(\th_M,\theta)\leq R(\th_F,\theta)+\epsilon.
\end{equation}
For each of the three definitions of approximate optimality, there is a corresponding notion of ``closeness" of quantum states such that if $\rho(\theta)$ is close to being classical, there is an approximately optimal measurement---approximately classical implies an approximately optimal measurement. We briefly review each of the results here, the first two of which are proved in Appendix \ref{sec:app_approx_opt}. The proofs are made up of a series of technical inequalities.

We start with additive optimality, considering closeness in trace norm: for any matrix $A$, define its trace norm as 
\begin{equation}
    \norm{A}_1 = \Tr(\sqrt{A^\dag A}).
\end{equation}
We denote $\sqrt{A^\dag A}$ as $|A|$.

We can now state the first result:
\begin{restatable}{lemma}{aao}
	Suppose $\rho(\theta)$ is some parametrised state and $\sigma(\theta)$ is some classical state. Suppose further that $\Theta$ has finite diameter $d=\sup_{\theta,\phi} L(\theta,\phi)$. If, for all $\theta$, we have that $\norm{\rho(\theta)-\sigma(\theta)}_1\leq \epsilon$, then there is a $2d\epsilon$-additively optimal measurement $M$. 
\end{restatable}

For the multiplicative error, we need a different notion of closeness, namely the maximum relative entropy \cite{Moso2013}: for quantum states $\rho$ and $\sigma$, we define
\begin{equation}
D_{\max}(\rho||\sigma) = \inf\{\gamma : \rho \leq e^\gamma \sigma\}.
\end{equation}

We can then state the multiplicative result:
\begin{restatable}{lemma}{amo}
	Suppose $\rho(\theta)$ is some parametrised state and $\sigma(\theta)$ is some classical state. If, for all $\theta$, we have that $\exp\{D_{\max}[\rho(\theta)||\sigma(\theta)]+ D_{\max}[\sigma(\theta)||\rho(\theta)]\}\leq 1+\eta$, then there is a $\eta$-multiplicatively optimal measurement $M$. 
\end{restatable}

The final closeness result is for local optimality. For a subset of real vectors $Y\subseteq\mathbb{R}^p$, denote its volume (Lebesgue measure) by $|Y|$.
\begin{lemma}
	Suppose $\rho(\theta)$ is some parametrised state, which is classical on some subset $\Gamma\subseteq\Theta$. Then there is a $(|\Theta|-|\Gamma|)$-locally optimal measurement $M$.
\end{lemma}
\begin{proof}
	Take the optimal measurement for $\rho|_\Gamma$. This clearly has the desired property.
\end{proof}

Note that it is not expected that being close to a classical state is necessary for an approximately optimal measurement to exist. For example, pure states $\ket{\psi(\theta)}$ where the quantum Fisher information can be saturated are expected to have an approximately optimal measurement for a large number of copies $\ket{\Psi(\theta)} = \ket{\psi(\theta)}^{\otimes n}, n\gg 1$. However, a large number of copies of a pure state is still a pure state and thus such states are far from being classical.

\section{Admissibility of Measurements}\label{sec:inad_of_meaus}
Our definition of when a measurement is better than another $(M\prec F)$ naturally gave rise to a definition of admissibility---a measurement $M$ is admissible if no other measurement dominates it. In this section we investigate two classes of intuitively ``bad" measurements and demonstrate their inadmissibility under some mild assumptions. 

To state these assumptions, we must introduce a specific class of loss functions---Bregman divergences \cite{Bane2005}. Specifically, for a convex set $\Theta\subseteq \mathbb{R}^p$ and any real-valued continuously differentiable, strictly convex function $g(\theta)$, the Bregman divergence associated to $g$ is given by
\begin{equation}
B(\theta_1,\theta_2) = g(\theta_1) - g(\theta_2) - [\grad g(\theta_2)]^T(\theta_1-\theta_2).
\end{equation}
A Bregman divergence is any function that can be written in such a fashion.
For example, least squares is a Bregman divergence with $g(\theta) = ||\theta||^2$ and the Kullback-Liebler divergence is a Bregman divergence with $g(\theta) = \sum_i \theta_i\log(\theta_i)$. Bregman divergences may be viewed as generalisations of the least-squares loss function.  Appendix \ref{sec:app_admis_of_meas} reviews two of their properties that we require for our results.

The first class of measurements we discuss are those that can be fine-grained to extract more information from a system. Suppose that one has some parameter estimation problem $\rho(\theta)$ and use some measurement $F$ where one knows the post measurement states. That is to say that $F=\{F_i^{\dag}F_i\}$ and, in the case of outcome $i$, the post measurement state is 
\begin{equation}
\rho_i = \frac{F_i\rho F_i^{\dag}}{\Tr(F_i\rho F_i^{\dag})}.
\end{equation}
Intuitively, if some $\rho_i$ still depends on $\theta$, then one can, in the case of outcome $i$, perform another measurement on the system to extract more information about the unknown parameter---refining our measurement. Recall that $\rho$ might be many copies of some density matrix: $\rho = \sigma^{\otimes n}$, so that outcome $i$ could be the result of many measurements of individual density matrices. More precisely, we say that $F$ is refineable if there exists an outcome (WLOG the first) as well as $\theta_1,\theta_2\in\Theta$ such that $p_1^{(F)}(\theta_1),\ p_1^{(F)}(\theta_2)> 0$ and $\rho_1(\theta_1)\neq\rho_1(\theta_2)$ \footnote{The first of these conditions is necessary to ensure that $\rho_1$ is well-defined at $\theta_1$ and $\theta_2$.}. In Appendix \ref{sec:app_admis_of_meas}, we prove the following Lemma, that shows that such measurements are indeed inadmissible:

\begin{restatable}{lemma}{rfin}\label{lem:refinable_inadmis}
	Let $\Theta\subseteq\mathbb{R}^N$, $L$ be some Bregman divergence loss function on $\Theta$ and $\rho(\theta)$ be some non-constant parametrised state. Suppose that $F=\{F_i^{\dag}F_i\}$ is a refineable measurement, then $F$ is inadmissible.
\end{restatable} 
\begin{proof}[Proof Sketch]
Constructing an $M$ such that $M\preceq F$ is relatively straightforward---in the case of the first outcome, our state still depends on the parameter $\theta$, so one can measure it again. We are always free to ignore the result of this subsequent measurement, and thus it is clear that $M\preceq F$. The proof that $F\npreceq M$ is somewhat technical and involves the framework of (classical) Bayesian estimation \cite{Lehm1998}, which is outlined in Appendix \ref{sec:Bayesian}.
\end{proof}

The second case of inadmissibility that we consider is when a measurement does not extract information from a system. For example, one could do a measurement with a single outcome or measure in a mutually unbiased basis to the basis our parameter is encoded in. Formally, this corresponds to the case where the measurement outcome probabilities are independent of the underlying parameter. Again, in Appendix \ref{sec:app_admis_of_meas}, we prove the following Lemma:

\begin{restatable}{lemma}{niin}\label{lem:const_prob_inadmis}
		Let $\Theta\subseteq\mathbb{R}^N$, $L$ be some Bregman divergence loss function on $\Theta$ and $\rho(\theta)$ be some non-constant parametrised state. Suppose that $F=\{F_k\}$ is a measurement whose outcome probabilities are independent of $\theta$, then $F$ is inadmissible under least-squares loss.
\end{restatable}
\begin{proof}[Proof Sketch]
	There are two main ideas in this proof. First, we show that the only admissible estimators for $F$ are constant estimators (described above) and thus any measurement $M$ satisfies $M\preceq F$. Second, we find an $M$ dominating $F$ using the Bayesian framework.
\end{proof}

Note both results Lemmas show that the identity measurement (i.e. doing nothing and just guessing) is admissible iff. the state $\rho$ is constant. Thus no matter how weakly a system varies, there is always a measurement that extracts useful information from it. 

\section{Conclusion}
In Conclusion, we have defined the most general way in which the performance of quantum measurements for parameter estimation can be compared. As we have argued, any other comparison would require additional modelling assumptions. Remarkably, we demonstrated a class of parametrised states---classical ones---that admit optimal measurements, even at this level of generality. Under some more strict assumptions, we have shown that these are the only such states. 

Since only a very restrictive class of parametrised states have optimal measurements, we proposed several criteria for when a measurement may be considered approximately optimal. A further direction of research that appears interesting is characterising when a state has an approximately optimal measurement---we have given a selection of sufficient conditions but, as argued, they are not necessary. In particular, understanding approximately optimal measurements in the asymptotic regime could provide a new perspective on well established ideas, such as the quantum Fisher information.

Finally, we have demonstrated two inadmissible classes of measurements---those that can be refined or those that do not extract information from the quantum state of interest. Another possible further direction of research would be to attempt to characterise the admissible estimators for a given parametrised quantum state.

\section*{Acknowledgments}
The authors wish to thank N. Mertig, J. Smith and C. Long for their useful discussions and comments. W.S. was supported by the EPSRC and Hitachi. S.S. acknowledges support from the Royal Society University Research Fellowship. D.R.M.A.-S. was supported by the Lars Hierta’s Memorial Foundation, and Girton College.

\printbibliography

\newpage
\appendix

\section{Bayesian estimation}\label{sec:Bayesian}
In this section we introduce Bayesian estimation, which we will need for the proofs of the results from the main text.
 
In a Bayes setting, one has some \textit{a priori} belief about our underlying parameter\cite{Wass2004}. For example, one may believe that it is likely to be centred around some value, and correspondingly very unlikely to be far away from it. Formally, this corresponds to a probability distribution $\pi(\theta)$ on our parameter space. A common choice is a normally distributed prior $\mathcal{N}(\mu,\Sigma)$; to encode an underlying expectation of our parameter $\mu$ along with decay of probability away from $\mu$ encoded by $\Sigma$. Since one has some probability distribution on our parameter space, one can now talk about the risk of an estimator averaged across different possible values of the parameter rather than at a specific parameter value. That is, one defines the Bayes risk corresponding to $\pi$ as

\begin{equation}
R_{\pi}(\th) = \mathbb{E}_{\theta\sim\pi(\theta)}[R(\th,\theta)] = \int d\theta \pi(\theta) R(\th,\theta).
\end{equation}

An estimator is $\tb$ is said to be Bayes if it minimises the Bayes risk, that is for any other estimator $\th$, $R_{\pi}(\tb) \leq R_{\pi}(\th)$. We note the following fact about Bayesian estimators.

\begin{lemma}\label{lem:uniq_bayes_ad}
	A unique Bayes estimator $\tb$ is admissible
\end{lemma} 
\begin{proof}
	Suppose $\th$ dominates $\tb$, then by definition $R(\th,\theta)\leq R(\tb,\theta)$ and $\th\neq\tb$. But then taking expectation with respect to $\theta$, we see that $R_{\pi}(\th) \leq R_{\pi}(\tb)$. But this contradicts that $\tb$ is the unique minimiser of the Bayes risk.
\end{proof}

In the quantum case, where one has a state $\rho(\theta)$, one must minimise the Bayes risk over estimators based off a fixed measurement and then additionally minimise this across all possible measurements. Thus one searches for a measurement-estimator pair which minimises the Bayes risk. In such a case, both the measurement and estimator are described as Bayes.

In the single parameter case $(\Theta\subseteq \mathbb{R})$, we also introduce some notation: for any measurement-estimator pair $(F, \th_F)$, we define
\begin{equation}
\Lambda = \sum_i F_i \th_F(i), \quad \Lambda_2 = \sum_i F_i \th_F(i)^2,
\end{equation}
Where $F_i$ are the measurement operators of $F$.

For any prior $\pi(\theta)$ on $\Theta$, we define two operators
\begin{equation}
\rb = \int d\theta\pi(\theta) \rho(\theta),\quad \rb' = \int d\theta\pi(\theta) \theta\rho(\theta).
\end{equation}

\section{Proof of Theorem \ref{thm:converse_to_optimality} }\label{sec:conv_to_opt}
In this section, we prove that only classical parametrised quantum states have optimal measurements. We specialise to the case that our parameter $\theta$ can be expressed as a real vector - $\Theta \subseteq \mathbb{R}^N$ and that we use least-squares loss - where $L(\theta_1,\theta_2) = ||\theta_1 - \theta_2||^2$. 

We now present the proof of Theorem \ref{thm:converse_to_optimality} as a series of Lemmas. We first restirct to the single parameter ($N=1$) case. We fix $\theta_1,\theta_2\in\Theta$ and consider the restricted parameter space $\{\theta_1,\theta_2\}\subseteq\Theta$ and show that $\rho(\theta_1)$ and $\rho(\theta_2)$ commute. This is achieved by considering the Bayesian measurements associated with all possible priors on $\{\theta_1,\theta_2\}$. This then shows that if an optimal measurement exists on $\Theta$, all $\{\rho(\theta):\theta\in\Theta\}$ must commute. We use this to simultaneously diagonalise the $\{\rho(\theta)\}$ and thus deduce $\rho$ is classical. By restricting to lines in higher dimension parameter spaces, we use the single parameter case to prove the multiparameter case.

\begin{lemma}\label{lem:optimal_for_Bayes}
	Suppose $\rho(\theta)$ is some parametrised state with an optimal measurement $M$. Then $M$ is a Bayes measurement for any prior $\pi(\theta)$ on $\Theta$.
\end{lemma}
\begin{proof}
	Suppose we have some Bayes estimator $\ttd^B_F$. Then since $M$ is optimal, $M\preceq F$, and thus we can find an estimator $\th_M$ satisfying $\th_M\leq \ttd^B_F$. But then, by definition, $\th_M$ must also be Bayes. Thus for any prior, we can always measure with $M$ and pick a Bayes estimator based on the resulting probability distribution.
\end{proof}

\begin{lemma}\label{lem:bayesian_conditions}
	Suppose $\rho(\theta)$ is some single parameter state $(\Theta\subseteq \mathbb{R})$ under least-squares loss. Suppose further we have some prior $\pi(\theta)$ on $\theta$ where $\rb$ has full rank. Then a measurement $F=\{F_i\}$, $F$ is Bayesian iff. the following hold
	\begin{enumerate}[label=(\roman*)]
		\item{
			$F$ is the projective measurement onto $\Lambda$'s eigenspaces, or a fine-graining of it.
		}
		\item{
			$\{\Lambda,\rb\} = 2\rb'$
		}
	\end{enumerate}
\end{lemma}
\begin{proof}
	Note that this lemma, and most of this proof, is presented in \cite{Demk2020} as a sufficient condition. We prove that the conditions are also necessary.
	
	Fix some measurement-estimator pair $(F,\th_F)$. Note that we may write the Bayesian risk in terms of $\Lambda$ and $\Lambda_2$ (in the case of least-squares loss):
	
	\begin{align}
	R_{\pi}(\th_F) &= \int d\theta \sum_i \pi(\theta) \Tr(F_i\rho(\theta)) (\th_F(i)-\theta)^2,\\
	&=  \Tr(\Lambda_2 \rb) -2 \Tr (\Lambda \rb') + \int d\theta \pi(\theta)\theta^2.
	\end{align}
	We note that
	\begin{align}
	0 &\leq \sum_i (\th(i)\mathbbm{1} - \Lambda)^\dag F_i(\th(i)\mathbbm{1} - \Lambda),\\
	&= \Lambda_2 - \Lambda^2,
	\end{align}
	and thus $\Tr(\Lambda_2 \rb) \geq \Tr (\Lambda^2 \rb)$. Since we want to minimise the Bayes risk, we now check when this inequality is saturated. Note that for any positive semi-definite operator $A$,
	\begin{align}
	\Tr(\rb A) = 0 &\Leftrightarrow \Tr((\sqrt{\rb}\sqrt{A})^\dag\sqrt{\rb}\sqrt{A}) = 0,\\
	&\Leftrightarrow \sqrt{\rb}\sqrt{A} = 0,\\
	&\Leftrightarrow A = 0,
	\end{align}
	where we used the assumption that $\rb$ is full rank, and hence $\sqrt{\rb}$, is invertible. Thus we saturate the inequality iff. $\sum_i (\th(i) - \Lambda)^\dag F_i(\th(i) - \Lambda)=0$. But since this is a sum of positive semi-definite operators, this can only happen if, for each $i$, $(\th(i) - \Lambda)^\dag F_i(\th(i) - \Lambda) = 0$. If $\th(i)$ is not an eigenvalue of $\Lambda$, then $F_i=0$. Otherwise, $F_i$ must only have support on the $\th(i)$ eigenspace of $\Lambda$. Since $\Lambda = \sum F_i \th(i)$, condition $(i)$ follows.
	
	On saturation of the inequality, we see that finding a Bayes measurement reduces to minimising the quantity
	\begin{equation}
	\mathcal{C} = \Tr(\Lambda^2 \rb) -2 \Tr (\Lambda \rb'),
	\end{equation}
	over Hermitian operators $\Lambda$. By differentiating with respect to $\Lambda$ and setting the derivative to zero, we see that there is a unique minimum and that it satisfies the equation 
	\begin{equation}
	\frac{1}{2} \{ \Lambda, \rb\} = \rb'.
	\end{equation}
\end{proof}

\begin{lemma}\label{lem:full_rank_case}
	Suppose $\rho(\theta)$ is some single parameter state $(\Theta\subseteq \mathbb{R})$ with an optimal measurement $M$ under least-squares loss. Fix $\theta_1,\theta_2\in\Theta$ distinct. If $\rho(\theta_1),\, \rho(\theta_2)$ both have full rank, then there exist simultaneously diagonalisable Hermitian maps $\Gamma, \Gamma'$ satisfying
	\begin{equation}
	\{\Gamma,\rho_2\}=\rho_1,\quad \{\Gamma', \rho_1\} = \rho_2.
	\end{equation}
\end{lemma}
\begin{proof}
	For notaional convenience, let $\rho_i=\rho(\theta_i)$ for $i=1,2$. For $p\in(0,1)$ fix a prior on $\Theta$ where we are given $\rho_1$ with probability $p$ and $\rho_2$ with probability $1-p$. Note that for any $p\in(0,1)$, $\rb$ has full rank. By Lemma \ref{lem:optimal_for_Bayes} we see that $M$ must be Bayesian for any value of $p$ - let $\th_M^p$ be a Bayes estimator for a fixed value of $p$. Then $\Lambda = \sum_i M_i \th_M^p(i)$ and $M$ must satisfy condition $(i)$ of Lemma \ref{lem:bayesian_conditions}. Then, as a function of $p$ we may expand $\Lambda=\sum_k \mu_k(p)\dyad{k}$ - in an eigenbasis that does not depend on $p$.
	
	By condition $(ii)$ of Lemma \ref{lem:bayesian_conditions}, for any value of $p\in(0,1)$ we know that $\Lambda$ satisfies
	\begin{equation}\label{eqn:optimal_lambda}
	\frac{1}{2} \{ \Lambda, p\rho_1 + (1-p)\rho_2 \} = p\theta_1\rho_1 + (1-p)\theta_2\rho_2.
	\end{equation}
	As well as being invertible, the smallest eigenvalue of $\rb$, $\lambda_{\min}$ is bounded below by the minimum of the eigenvalues of $\rho_1$ and $\rho_2$.
	Consider equation \eqref{eqn:optimal_lambda} in the limit of $p\to 0$. Let $\Gamma$ be defined implicitly as the solution to the equation
	\begin{equation}\label{eqn:def_gamma}
	\{\Gamma,\rho_2\} = \rho_1.
	\end{equation}
	Note since $\rho_1$ is invertible, by expanding equation \eqref{eqn:def_gamma} in the eigenbasis of $\rho_2$, we see that $\Gamma$ is well defined and is furthermore Hermitian.
	
	For some ``remainder" matrix $R$, we expand $\Lambda$ as $\Lambda = \theta_2 \mathbbm{1} + 2p(\theta_1-\theta_2)\Gamma + R$. Substituting this ansatz into equation \eqref{eqn:optimal_lambda}, we see that $R$ must satisfy the equation
	\begin{equation}
	\frac{1}{2}\{R,\rb\} = p^2(\theta_2-\theta_1)\{\Gamma,\rho_1-\rho_2\}.
	\end{equation}
	Expanding in an eigenbasis of $\rb$ and recalling that all the eigenvalues of $\rb$ are bounded below by $\lambda_{\min}$ (which is independent of $p$), we see that $R=\mathcal{O}(p^2)$.
	Consider an eigenstate $\ket{\psi}$ of $\Lambda$ with eigenvalue $\mu(p)$, then note that
	\begin{equation}
	\norm{\Gamma\ket{\psi} - \frac{\mu(p)-\theta_2}{2p(\theta_1-\theta_2)}\ket\psi} = \mathcal{O}(p)
	\end{equation}
	Taking the limit as $p\to 0$, we see that $\ket\psi$ must be an eigenstate of $\Gamma$. But, by symmetry, in the limit $p \to 1$, we may define $\Gamma'$ by the equation
	\begin{equation}\label{eqn:def_gamma'}
	\{\Gamma',\rho_1\} = \rho_2,
	\end{equation}
	and $\ket\psi$ must also be an eigenstate of $\Gamma'$. Thus $\Gamma$ and $\Gamma'$ must be simultaneously diagonalisable in $\Lambda$'s eigenbasis. 
\end{proof}

\begin{lemma}\label{lem:solutions_commute}
	Suppose $A,B$ are positive definite and $\Gamma, \Gamma'$ are simultaneously diagonalisable Hermitian maps satisfying
	\begin{equation}
	\{\Gamma,A\}=B,\quad \{\Gamma', B\} = A.\label{eqn:anti_system}
	\end{equation}
	Then $A$ and $B$ commute.
\end{lemma}
\begin{proof}
	Let $\Gamma$ and $\Gamma'$ be simultaneously diagonalisable in the orthonormal basis $\{\ket i\}_{i=1}^n$. Let $\Gamma \ket i = \lambda_i \ket i$, $\Gamma' \ket i = \mu_i \ket  i$, $\bra{i}A\ket{j} = A_{ij}$ and $\bra{i}B\ket{j} = B_{ij}$, for $i,j=1,\dots,n$. By positive definiteness, note that $A_{ii},B_{ii}>0$ for any $i$.
	Rewriting equations \eqref{eqn:anti_system} in components of this basis, we reach the set of equations
	\begin{equation}\label{eqn:comps}
	(\lambda_i + \lambda_j)A_{ij} = B_{ij},\quad  (\mu_i + \mu_j)B_{ij} = A_{ij}.
	\end{equation}
	Setting $i=j$ we deduce that 
	\begin{equation}\label{eqn:evals}
	\lambda_i = B_{ii}/2A_{ii}, \quad \mu_i = A_{ii}/2B_{ii}.
	\end{equation}
	Note equation \eqref{eqn:comps} shows $A_{ij} = 0 \Leftrightarrow B_{ij} = 0$. Suppose $A_{ij}\neq 0$, then multiplying the two individual equations in \eqref{eqn:comps} and dividing by $A_{ij}B_{ij}$ we see that
	\begin{align}
	&\left(A_{ii}/2B_{ii} +A_{jj}/2B_{jj}\right)(B_{ii}/2A_{ii} + B_{jj}/2A_{jj}) = 1 \nonumber\\
	&\Leftrightarrow \frac{A_{ii}B_{jj}}{B_{ii}A_{jj}} + \frac{B_{ii}A_{jj}}{A_{ii}B_{jj}} = 2.
	\end{align}
	But $x+1/x = 2$ iff. $x=1$ and thus we deduce 
	\begin{equation}\label{eqn:ratio}
	\frac{A_{ii}}{B_{ii}} = \frac{A_{jj}}{B_{jj}}
	\end{equation}
	
	Consider a graph $G$ on $n$ nodes, with an edge between $i$ and $j$ iff. $A_{ij}\neq0$. If $G$ is connected, then for any nodes $i,j$ in $G$ there is path between them. Assuming connectedness, we apply the result of equation \eqref{eqn:ratio} along the path, we deduce that $A_{ii}B_{jj}=B_{ii}A_{jj}$ for every $i,j=1,\dots,n$. But then summing over $j$, we deduce that $A_{ii} = \frac{\Tr(A)}{\Tr(B)}B_{ii}$. Substituting this into equations \eqref{eqn:comps} and \eqref{eqn:evals}, we see that $A$ and $B$ are proportional and thus commute.
	If $G$ is not connected, then we can apply the above procedure to each connected component of $G$. The result is that $A$ and $B$ are block diagonal, with the diagonal matrix entries proportional and thus $A$ and $B$ commute.	
\end{proof}

\begin{lemma}\label{lem:states_commute}
	Suppose $\rho(\theta)$ is some single parameter state $(\Theta\subseteq \mathbb{R})$ with an optimal measurement $M$ under least-squares loss. Fix $\theta_1,\theta_2\in\Theta$ distinct, then $\rho(\theta_1),\, \rho(\theta_2)$ commute (regardless of their rank).
\end{lemma}
\begin{proof}
	Again, for notational convenience, let $\rho_i=\rho(\theta_i)$ for $i=1,2$. Redefine our parameter space $\Theta = \{\theta_1,\theta_2\}$, but still allow estimators to take any real value. Note, by definition, that $M$ must also be optimal for this parameter estimation problem.
	
	Suppose that $U=\ker(\rho_2)\cap \ker(\rho_1)\neq 0$. Let $U^\perp$ be the orthogonal compliment of $U$ and let $\Pi_{U^\perp}$ be the orthogonal projection matrix onto $U^\perp$. Then note that replacing $M$ with the measurement $\{\Pi_{U^\perp} M_i \Pi_{U^\perp}, \mathbbm{1} - \Pi_{U^\perp} \}$ does not chance any of the measurement probabilities when measuring $\rho_1$ or $\rho_2$. Thus we may WLOG restrict to $U^\perp$, which still has an optimal measurement for this restricted parameter estimation problem.
	
	Suppose that $\rho_2$ is not full rank, i.e. it has some non-trivial kernel $K\leq\mathcal{H}$. Take the measurement $E$ of the two orthogonal projectors $\Pi_K$ and $\Pi_\Kperp$, along with an estimator $\th_E(K) = \theta_1$, $\th_E(\Kperp) = \theta_2$ so that
	\begin{align}
	R(\th_E,\theta_2) &= 0,\\
	R(\th_E,\theta_1) &= (\theta_1-\theta_2)^2(1-\Tr(\Pi_K \rho_1)).
	\end{align} 
	Since $M$ is optimal, we have that $M\preceq E$. Then, in particular, there must be an estimator $\ttd_M$ satisfying $\ttd_M\leq\th_E$. In particular, $\ttd_M$ must have zero risk at $\theta_2$ and thus if $\Tr(\rho_2 M_i)\neq0$, then we must have $\ttd_M(i)=\theta_2$. Let $I=\{i\ |\ \Tr(\rho_2M_i) = 0\}$, then, by the above, $\ttd_M$ must satisfy
	\begin{equation}\label{eqn:ttd_risk}
	R(\ttd_M,\theta_1) \geq (\theta_1-\theta_2)^2 \left(1-\sum_{i\in I}\Tr(\rho_1M_i)\right),
	\end{equation}
	with equality iff. for every $i\in I$, $\ttd_M(i)=\theta_1$. WLOG Assume that this holds, so that we saturate equation \eqref{eqn:ttd_risk}. If $\Tr(\rho_2M_i) = 0$, then, as $M_i\geq 0$, every eigenvector of $M_i$ with non-zero eigenvalue must lie in $K$ and thus $\sum_{i\in I} M_i \leq \Pi_K$. Thus, $R(\ttd_M,\theta_1)\geq R(\th_E,\theta_2)$ with equality iff. $\rho_1\left(\Pi_K - \sum_{i\in I} M_i\right)=0$. Since $\rho_1$ does not kill any vector in $K$ (as we restricted to $U^\perp$) we get equality iff. $\sum_{i\in I} M_i = \Pi_K$. Thus for $i\notin I$ and $\ket{k}\in K$, $M_i\ket{k}=0$.
	
	Now, for $p\in(0,1)$ fix a prior on $\Theta$ where we are given $\rho_1$ with probability $p$ and $\rho_2$ with probability $1-p$. Taking $\Lambda$ corresponding to $M$ by Lemma \ref{lem:optimal_for_Bayes}, we see that $\Lambda$ decomposes as $\Lambda = \Lambda_K + \Lambda_\Kperp$, where $\Lambda_K = \Pi_K \Lambda_K \Pi_K$ and $\Lambda_\Kperp = \Pi_\Kperp \Lambda_\Kperp \Pi_\Kperp$. Furthermore, in order for our estimator to be Bayes, we must always guess $\theta_1$ in the case of a $K$ outcome i.e. $\Lambda_K = \theta_1 \Pi_K$.
	
	Take the inner product of condition $(ii)$ of lemma \ref{lem:bayesian_conditions} with $\ket{k}\in K$ and $\ket{\ell}\in \Kperp$. This gives
	\begin{equation}\label{eqn:kernel_fixed}
	\bra{\ell} \Lambda_\Kperp \rho_1 \ket{k} = \theta_1\bra{\ell}\rho_1\ket{k}.
	\end{equation}
	Note that as $p\to 0$, we must have $\Lambda_\Kperp \to \Pi_\Kperp \theta_2$, as our estimates must approach $\theta_2$ for our estimator to be Bayes. Thus the only way for equation $\eqref{eqn:kernel_fixed}$ to be satisfied for all $p\in(0,1)$ is for $\bra{\ell}\rho_1\ket{k}=0$. Thus $\rho_1$ fixes $K$ and we may decompose it as $\rho_1 = \Pi_K\rho_1\Pi_K + \Pi_\Kperp\rho_1\Pi_\Kperp$. Thus equation \eqref{eqn:optimal_lambda} becomes
	\begin{align}\label{eqn:reduced_lambda}
	&\frac{1}{2} \{ \Lambda_\Kperp, p \Pi_\Kperp \rho_1 \Pi_\Kperp  + (1-p)\Pi_\Kperp \rho_2\Pi_\Kperp\} \nonumber\\
	&= p\theta_1\Pi_\Kperp \rho_1\Pi_\Kperp  + (1-p)\theta_2\Pi_\Kperp \rho_2\Pi_\Kperp .
	\end{align}
	This essentially reduces the Hilbert space to $\Kperp$, on which $\rho_2$ has full rank. We may then repeat the above to assume $\rho_1$ has full rank (note that, as we restricted to $U^\perp$, $\ker(\rho_1)\leq \Kperp$. By Lemmas \ref{lem:full_rank_case} and \ref{lem:solutions_commute}, $\rho_1$ and $\rho_2$ commute on their joint support. But as they fix each other's kernels and are Hermitian, they must, therefore, fully commute.
\end{proof}

\begin{lemma}\label{lem:single_param_case}
	Suppose $\rho(\theta)$ is some single parameter state $(\Theta\subseteq \mathbb{R})$. If there is an optimal measurement $M$ under least-squares loss, then $\rho(\theta)$ is a classical state problem.
\end{lemma}
\begin{proof}
	Fix $\theta_1$ in $\Theta$ and decompose $\mathcal{H}$ as a direct sum of $\rho(\theta_1)$'s (distinct) eigenspaces $\mathcal{H} = \oplus_j V_j$. If this is an eigenspace decomposition of $\rho(\theta)$ for every $\theta\in\Theta$, then by definition $\rho$ is a classical state. Otherwise, take $\theta_2$ such that some $V_j$ is not an eigenspace of $\rho(\theta_2)$. But by Lemma \ref{lem:states_commute}, $[\rho(\theta_1),\rho(\theta_2)]=0$ and thus $V_j$ must decompose a sum of eigenspaces of $\rho(\theta_2)$, $V_j = \oplus V'_j$. Take this new decomposition into eigenspaces. We may repeat this process, noting that it must terminate as each eigenspace must have dimension at least one, to see that $\rho(\theta)$ must be classical. 
\end{proof}

\cto*
\begin{proof}
	Suppose we have a general $\rho(\theta)$ and some optimal measurement $M$. Fix $\theta_1,\theta_2\in\Theta$ and let $T=[0,||\theta_2-\theta_2||]$.By convexity, for $t\in T$, $\gamma(t) = \theta_1 + t(\theta_2 - \theta_1)/||\theta_2 - \theta_1||$ is in $\Theta$. Thus $\rho(\gamma(t))$ gives a single parameter estimation problem - estimating $t$. We have constructed this problem such that the least-squares loss functions agree - that is for $s,t\in T$,
	\begin{equation}
	||\gamma(t)-\gamma(s)||^2 = |t-s|^2.
	\end{equation} 
	But then (by projecting any estimators $\th_M$ onto the line segment between $\theta_1$ and $\theta_2$) $M$ must also be optimal for this problem, and thus, by Lemma \ref{lem:single_param_case}, $\rho(\theta_1)$ and $\rho(\theta_2)$ commute. But then, since $\theta_1$ and $\theta_2$ were arbitrary, by the same argument as Lemma \ref{lem:single_param_case}, we see that $\rho(\theta)$ must be classical.
\end{proof}

We remark that the the condition of convexity in Theorem \ref{thm:converse_to_optimality} can be weakened. By Lemma \ref{lem:bayesian_conditions}, we see that the optimal measurement must be a fine-graining of the projection onto $\rho$'s joint eigenspaces. But then for any two distinct sets on which $\rho$ is classical, we see that it must be classical on the union of these sets too. Then the result holds for $\Theta$ a disjoint union of convex sets too. Moreover, if $\Theta$ is open, around any $\theta\in\Theta$ we can fit a convex set in which, $\rho$ must be classical. But then by the same reasoning as before, we see that $\rho$ must be classical on the whole of $\Theta$. We do not prove this result in full detail, as most parameter spaces of interest are convex.

\section{Approximately Classical Implies an Approximately Optimal Measurement}\label{sec:app_approx_opt}
In this section we provide the remaining two proofs for each of the ``close" to classical implies ``close" to optimal results from Section \ref{sec:approx_opt}. It will be useful to slightly extend our notation for risk functions to include a label for the state we are considering, i.e. we write $R_\rho(\th_M,\theta)$.

We start with the first result:

\aao*
\begin{proof}
	Since $\sigma(\theta)$ is classical, we can fix some optimal measurement $M$. Fix some measurement-estimator pair $(F,\th_F)$ and $\theta\in\Theta$. Then note that
	\begin{align}
	|R_\rho(\th_F,&\theta)-R_\sigma(\th_F,\theta)| \nonumber\\
	&= \bigg|\sum_i \Tr[F_i (\rho(\theta)-\sigma(\theta))]L(\th_F(i),\theta) \bigg|,\\
	&\leq d\sum_i |\Tr[F_i (\rho(\theta)-\sigma(\theta))]|.\label{eqn:eps_ineq}
	\end{align}
	Fix some $F_i\geq 0$. Diagonalising $\rho(\theta)-\sigma(\theta)= \sum_j  \lambda_j\dyad{j}$, note that
	\begin{align}
	|\Tr[F_i (\rho(\theta)-\sigma(\theta))]| &\leq \sum_j |\lambda_j| \bra{j}F_i\ket{j},\\
	&= \Tr( F_i|\rho(\theta)-\sigma(\theta)|).\label{eqn:trace_ineq}
	\end{align}
	Substituting inequality \eqref{eqn:trace_ineq} into \eqref{eqn:eps_ineq}, we see that 
	\begin{equation}
	|R_\rho(\th_F,\theta)-R_\sigma(\th_F,\theta)| \leq d\epsilon \label{eqn:risk_ineq}
	\end{equation}
	Since $M$ is optimal on $\sigma(\theta)$ there exists $\th_M$ such that for all $\theta$, $R_\sigma(\th_M,\theta)\leq R_\sigma(\th_F,\theta)$. Then applying \eqref{eqn:risk_ineq} twice, we see
	\begin{align}
	R_\rho(\th_M,\theta) - R_\rho(\th_F,\theta) &\leq R_\sigma(\th_M,\theta)-R_\sigma(\th_F,\theta) + 2d\epsilon,\\
	&\leq 2d\epsilon.
	\end{align}
\end{proof}

For the multiplicative error, we will make use of the following property of maximum relative entropy,
\begin{lemma}\label{lem:meaus_max_true}
	For two states $\rho,\sigma$
	\begin{equation}
	D_{\max}(\rho||\sigma) = \log\sup_{0\leq M\leq \mathbbm{1}}  \frac{\Tr(M\rho)}{\Tr(M\sigma).}
	\end{equation}
\end{lemma}
We omit the proof, it is given in \cite{Moso2013}.

Using this, we can prove the closeness result for multiplicative approximate optimality
\amo*
\begin{proof}
	Fix some measurement $F$, note that
	\begin{align}
	\frac{R_\rho(\th_F,\theta)}{R_\sigma(\th_F,\theta)} &= \frac{\sum _i \Tr(\rho(\theta)F_i)L(\th_F(i,\theta))}{\sum _i \Tr(\sigma(\theta)F_i)L(\th_F(i,\theta))},\\
	&\leq \max_i\frac{\Tr(\rho(\theta)F_i)}{\Tr(\sigma(\theta)F_i)},\\
	&\leq e^{D_{\max}(\rho || \sigma)},\label{eqn:mul_result}
	\end{align}
	where we have used Lemma \ref{lem:meaus_max_true}. But then fixing an optimal measurement for $\sigma$ and using the inequality \eqref{eqn:mul_result} twice, the result follows.
\end{proof}

\section{Admissibility of Measurements}\label{sec:app_admis_of_meas}
The aim of this section is to prove that the two classes of measurements discussed in Section \ref{sec:inad_of_meaus} are inadmissible.

To begin, we must prove a series of technical results about Bregman divergences and Bayesian estimation (see Appendix \ref{sec:Bayesian}).There are two properties of Bregrman divergences that we will need, stated below in Lemma \ref{lem:Breg_props}. We will not prove them, instead referring the reader to \cite{Bane2005}.

\begin{lemma}\label{lem:Breg_props}
	Let $L$ a Bregman divergence. Then 
	\begin{enumerate}[label=(\roman*)]
		\item{
			L is strictly convex in its first argument.
		}
		\item{
			For any prior $\pi$ on $\Theta$, the Bayes estimator $\tb$ is unique and is given by the posterior mean $\tb(x) = \mathbb{E}_{\pi}[\theta|\text{outcome } x]$.
		}
	\end{enumerate}
\end{lemma}

Next we prove some technical lemmas to do with Bayesian estimation.

\begin{lemma}\label{lem:non-const_estim}
	Let $\Theta\subset\mathbb{R}^p$ and $L$ be a Bregman divergence. Suppose that $\{p_{i}(\theta)\}$ is a distribution depending on $\theta$ and that there exist $\theta_1,\theta_2\in\Theta$ such that $p_1(\theta_1)p_2(\theta_2)\neq p_2(\theta_1)p_1(\theta_2)$. Then, there exists an admissible estimator $\th$ where $\th(1)\neq\th(2)$.
\end{lemma}
\begin{proof}
	Define a Prior on $\Theta$ of uniformly random choice of $\theta_1$ or $\theta_2$, that is to say,
	\begin{equation}
	\mathbb{P}(Y=\theta) = \begin{cases}
	1/2,	&\theta\in\{\theta_1,\theta_2\},\\
	0,	& \text{o.w.}
	\end{cases}
	\end{equation}
	In this case, by Lemma \ref{lem:Breg_props}, we know that 
	\begin{align}
	\tb(i) &= \mathbb{E}_{\pi}[\theta|i],\\
	&= \mathbb{P}(\theta_1|i)\theta_1 + \mathbb{P}(\theta_2|i)\theta_2. 
	\end{align}
	Since this is a convex combination of $\theta_1$ and $\theta_2$ which must be distinct to satisfy the conditions of the Lemma, we deduce that $\tb(1)=\tb(2)$ iff. $\mathbb{P}(\theta_1| 1) = \mathbb{P}(\theta_1 | 2)$.
	But
	\begin{equation}
	\mathbb{P}(\theta_1|i) = \frac{p_i(\theta_1)}{p_i(\theta_1)+p_i(\theta_2)}.
	\end{equation}
	So
	\begin{align}
	\mathbb{P}(\theta_1| 1) = \mathbb{P}(\theta_1 | 2) &\Leftrightarrow \frac{p_1(\theta_1)}{p_1(\theta_1)+p_1(\theta_2)} = \frac{p_2(\theta_1)}{p_2(\theta_1)+p_2(\theta_2)}, \\
	&\Leftrightarrow p_1(\theta_1)p_2(\theta_2) = p_2(\theta_1)p_1(\theta_2).
	\end{align}
	Thus by the assumptions of the Lemma, $\tb(1)\neq\tb(2)$. But by Lemma \ref{lem:uniq_bayes_ad},
	$\tb$ is admissible.
\end{proof}

\begin{lemma}\label{lem:same_inadmiss}
	Let $L$ be a Bregman divergence. Suppose that $\{p_{i}(\theta)\}$ is a distribution depending on $\theta$ and that $\th,\ttd$ are two distinct estimators with the same risk function. Then $\th$ and $\ttd$ are inadmissible unless they are equal with probability 1.
\end{lemma}
\begin{proof}
	Let $\th'(i) = [\th(i)+\ttd(i)]/2$. By strict convexity of $L$ in its first argument and since $\th,\ttd$ are distinct, the result follows.
\end{proof}

\begin{lemma}\label{lem:const_admissible}
	Let $L$ be a Bregman divergence and $\{p_i\}$ be some distribution independent of some parameter $\theta$. Then an estimator $\th$ is admissible iff. it is constant with probability 1.
\end{lemma}
\begin{proof}
	Consider the constant estimator $\theta_0 = \sum_i p_i\th(i)$, which is well defined by the assumption that $p_k^{(F)}$ is constant. Then
	\begin{align}
	R(\th,\theta) &= \sum_i p_i L(\th(i),\theta)\\
	&\geq L(\theta_0,\theta) = R(\theta_0,\theta). \label{eqn:const_est}
	\end{align}
	By lemma \ref{lem:same_inadmiss}, $\th$ is admissible only if it equals $\theta_0$ with probability 1.
\end{proof}

We can now prove the two Lemmas from Section \ref{sec:inad_of_meaus}. First, we deal with refineable measurements.  Recall their definition: suppose that we have some parameter estimation problem $\rho(\theta)$ and a measurement $F=\{F_i^{\dag}F_i\}_{i=1}^I$. In the case of outcome $i$ the post measurement state is
\begin{equation}
\rho_i = \frac{F_i\rho F_i^\dag}{\Tr(F_i\rho F_i^\dag)}.
\end{equation}
Recall that a measurement $F=\{F_i^{\dag}F_i\}_{i=1}^I$ is called refineable if one of the possible post measurement states (WLOG the first) still depends on $\theta$. That is there are $\theta_1,\theta_2\in\Theta$ such that
\begin{enumerate}[label=(\roman*)]
	\item{
		$\rho_1(\theta_1)\neq\rho_1(\theta_2).$
	}
	\item{
		$p_1(\theta_1),\, p_1(\theta_2)\neq 0$.
	}
\end{enumerate}
The second of these conditions is to ensure that outcome one is possible at $\theta_1$ and $\theta_2$. Otherwise, the post measurement state, as defined above, is not well-defined. We can now prove the first result on inadmissibility from the main text.

\rfin*
\begin{proof}
	Take $\theta_1,\theta_2$ satisfying the conditions above. Since $\rho_1(\theta_1)\neq\rho_1(\theta_2)$, by the Helstrom Bound \cite{Hels1969}, there is some measurement $M={M_1,M_2}$ where $p_1^{(M)}(\theta_1)\neq p_1^{(M)}(\theta_2)$. Consider the measurement $MF$ corresponding to measuring first with $F$ and in the case of outcome 1, measuring again but now using $M$. That is $MF=\{F_1^{\dag}M_1^{\dag}M_1F_1,F_1^{\dag}M_2^{\dag}M_2F_1\}\cup \{F_i^{\dag}F_i\}_{i=2}^I$. We label the outcomes of this measurement as $(1,1),(2,1)$ and $i$ for $i=2,\dots,I$ in the obvious way. Note that for any estimator $(F,\th)$ we can construct an estimator $(MF, \th')$ by $\th'(1,1)=\th'(1,2) = \th(1)$ and for $i\geq 2$, $\th'(i)=\th(i)$. By construction $R(\th,\theta)\equiv R(\th',\theta)$ and thus $MF\preceq F$. It remains to show $F\nprec MF$, which we show by contradiction. 
	
	Suppose we find an $MF$-admissible estimator $\th$ where $\th(1,1)\neq\th(1,2)$. If there is an estimator $(F, \ttd)$ such that $(F, \ttd)\leq(MF, \th)$, then by the above construction $(MF,\ttd')\leq(MF,\th)$. But by construction, $\ttd'\neq\th$, and so by lemma \ref{lem:same_inadmiss} this contradicts admissibility of $\th$ and hence $F\nprec MF$. Thus the problem reduces to finding such an estimator $\th$.
	
	Note that by lemma \ref{lem:non-const_estim} it is sufficient to show $p_{(1,1)}^{(MF)}(\theta_1)p_{(1,2)}^{(MF)}(\theta_2)\neq p_{(1,2)}^{(MF)}(\theta_1)p_{(1,1)}^{(MF)}(\theta_2)$. But by expanding probabilities and dividing through by $p_1^{(F)}(\theta_1)p_1^{(F)}(\theta_2)$, we see that
	\begin{align}
	&p_{(1,1)}^{(MF)}(\theta_1)p_{(1,2)}^{(MF)}(\theta_2)= p_{(1,2)}^{(MF)}(\theta_1)p_{(1,1)}^{(MF)}(\theta_2),\\
	&\Leftrightarrow p_1^{(M)}(\theta_1)p_2^{(M)}(\theta_2)=p_2^{(M)}(\theta_1)p_1^{(M)}(\theta_2),\\
	&\Leftrightarrow p_1^{(M)}(\theta_1) = p_1^{(M)}(\theta_2). \label{eqn:final_eq}
	\end{align}
	But by the definition of $M$ the final equality \eqref{eqn:final_eq} cannot hold. Thus we can apply lemma \ref{lem:non-const_estim} and a desired $\th$ exists.
\end{proof}

Next, we consider the case when a measurement extracts no information from a system.
\niin*
\begin{proof}
	Since $\rho$ is not constant, take $\theta_1,\theta_2\in\Theta$ such that $\rho(\theta_1)\neq\rho(\theta_2)$. Again, by the Helstrom bound, \cite{Hels1969}
	, there is some measurement $M=\{M_1, M_2\}$ where $p_1^{(M)}(\theta_1)\neq p_1^{(M)}(\theta_2)$. Then, note that
	\begin{align}
	&p_1^{(M)}(\theta_1)p_2^{(M)}(\theta_2)=p_2^{(M)}(\theta_1)p_1^{(M)}(\theta_2)\\ &\Leftrightarrow p_1^{(M)}(\theta_1) = p_1^{(M)}(\theta_2).
	\end{align}
	So by lemma \ref{lem:non-const_estim} there exists a non-constant admissible estimator $\th$. Applying Lemma \ref{lem:const_admissible}, we see $M\preceq F$ but $F\npreceq M$.
\end{proof}

\end{document}